\documentclass{article}

 \usepackage[nonatbib,final]{nips_2016}

\usepackage[utf8]{inputenc} 
\usepackage[T1]{fontenc}    
\usepackage[hidelinks]{hyperref}       
\usepackage{url}            
\usepackage{booktabs}       
\usepackage{amsfonts}       
\usepackage{nicefrac}       
\usepackage{microtype}      

\usepackage{amsmath,amssymb,bm,dsfont,enumitem,amsthm}
\usepackage{color}
\usepackage{xfrac}
\usepackage{epstopdf}
\usepackage{algorithm}
\usepackage[noend]{algpseudocode}

\newcommand{\wsig}{q}
\newcommand{\Wsig}{Q}

\newcommand{\nosig}{\bar{q}}
\newcommand{\Nosig}{\bar{Q}}

\newcommand{\needcite}[1]{}

\newcommand{\be}{\begin{equation}}
\newcommand{\ee}{\end{equation}}
\newcommand{\benn}{\begin{equation*}}
\newcommand{\eenn}{\end{equation*}}
\newcommand{\bea}{\begin{eqnarray*}}
\newcommand{\eea}{\end{eqnarray*}}
\newcommand{\bean}{\begin{eqnarray}}
\newcommand{\eean}{\end{eqnarray}}

\newcommand{\bb}{\boldsymbol{b}} 
\newcommand{\cc}{\boldsymbol{c}}
\newcommand{\ev}{\boldsymbol{e}}
\newcommand{\piv}{\boldsymbol{\pi}}

\newcommand{\rr}{\boldsymbol{r}}

\newcommand{\ignore}[1]{}

\renewcommand{\eqref}[1]{Eq.~(\ref{#1})}

\newcommand{\secref}[1]{Sec. \ref{#1}}

\newtheorem{theorem}{Theorem}[section]

\newtheorem{corollary}[theorem]{Corollary}

\newcommand{\commentout}[1]{}

\newcommand{\red}[1]{}
\newcommand{\blue}[1]{}

\newcommand{\darkgray}[1]{{\textcolor[RGB]{120,120,120}{#1}}}
\newcommand{\todo}[1]{}

\DeclareMathOperator*{\argmax}{argmax}

\newcommand{\inner}[1]{\langle {#1} \rangle}

\renewcommand{\k}{{(k)}}
\newcommand{\kk}{{(k+1)}}
\newcommand{\sub}{S}
\newcommand{\itm}{\omega}
\newcommand{\itms}{\Omega}
\newcommand{\tg}{\tau}
\newcommand{\tgs}{T}
\newcommand{\1}[1]{\mathds{1}_{\{{#1}\}}}
\newcommand{\prob}{\mathbb{P}}

\newcommand{\probargsub}[2]{{\prob}_{#2}\left[{#1}\right]}

\makeatletter
\algrenewcommand\ALG@beginalgorithmic{\normalsize}
\algrenewcommand\algorithmiccomment[2][\small]{{#1\hfill\(\triangleright\) \darkgray{#2}}}
\makeatother

\title{Optimal Tagging with Markov Chain Optimization}

\author{
Nir Rosenfeld\\
  School of Computer Science and Engineering\\
  Hebrew University of Jerusalem\\
  \texttt{nir.rosenfeld@mail.huji.ac.il} \\
  \And
 Amir Globerson \\
 The Blavatnik School of Computer Science  \\
 Tel Aviv University \\
 \texttt{gamir@post.tau.ac.il} \\
}

\begin{document}

\maketitle

\begin{abstract}
Many information systems use tags and keywords to describe and annotate content. These allow for efficient organization and categorization of items, as well as facilitate relevant search queries. As such, the selected set of tags
 for an item can have a considerable effect on the volume of traffic that eventually reaches an item.

In settings where tags are chosen by an item's creator, who in turn is interested in maximizing traffic, a principled approach for choosing tags can prove valuable. In this paper we introduce the problem of optimal tagging, where
 the task is to choose a subset of tags for a new item such that the probability of a browsing user reaching that item is maximized.

We formulate the problem by modeling traffic using a Markov chain, and asking how transitions in this chain should be modified to maximize traffic into a certain state of interest. The resulting optimization problem
involves maximizing a certain function over subsets, under a cardinality constraint.

We show that the optimization problem is NP-hard, but nonetheless has a simple
$(1-\frac{1}{e})$-approximation via a simple greedy algorithm. Furthermore, the structure of the problem allows for an efficient implementation of the
greedy step.
To demonstrate the effectiveness of our method, we perform experiments on three  tagging datasets, and show that the greedy algorithm outperforms other baselines.
\end{abstract}


\section{Introduction}
\label{sec:intro}

To allow for efficient navigation and search,
modern information systems rely on the usage of 
non-hierarchical tags, keywords, or labels
to describe items and content.
These tags are then used either explicitly by users when searching for content,
or implicitly by the system to augment search results or to
recommend related items.

Many online systems where users can create or upload content support tagging.
Examples of such systems are
media-sharing platforms,
social bookmarking websites,
and consumer to consumer auctioning services.
Typically, an item's creator is free to select any set of tags or keywords which 
they believe best describe their item, where the only concrete limitation
is on the number of tags, words, or characters used.
Tags are often chosen on a basis of their ability to best describe,
classify, or categorize items and content.
By choosing relevant tags, users aid in creating a more organized information system.
However, content creators may have their own individual objective,
such as maximizing the exposure of other users to their items.

This suggests that choosing tags should in fact be done strategically.
For instance, tagging a song as `Rock' may be informative,
but will probably only contribute marginally to the song's traffic,
as the competition for popularity under this tag can be fierce.
On the other hand, choosing a unique, obscure tag may be appealing,
but will probably not help much either.
Strategic tagging or keyword selection is more clearly exhibited
in search-supporting systems, where keywords are explicitly used
for filtering and ordering search results or ad placements,
and users have a clear incentive of maximizing the exposure of their item.
Nonetheless, their selections are typically heuristic or manual.

Recent years have seen an abundance of work on methods for user-specific
tag recommendations
\cite{hotho2006folkrank,fang2015personalized}.
Such methods aim to support collaborative tagging systems,
where any user can tag any item in the repository.
We take a complementary perspective and focus on taxonomic tagging systems
where only the creator of an item can determine its tags.
In this paper we formalize the task of \emph{optimal tagging} and suggest an
efficient, provably-approximate algorithm.
While the problem is shown to be NP-hard, we prove that the objective
is in fact monotone and submodular, which suggests a straightforward
greedy $(1-\frac{1}{e})$-approximation algorithm \cite{nemhauser1978analysis}.
We also show how the greedy step, which consists of solving a
set of linear equations, can be greatly simplified as well as parallelized,
resulting in a significant improvement in runtime.

We begin by modeling a user browsing a tagged information system as a random walk.
Items and tags act as states in a Markov chain, whose transition probabilities
describe the probability of users jumping between items and tags.
Our framework can incorporate many web search click models
\cite{chuklin2015click}.
Given a new item, our task is to choose a subset of $k$ tags for this item.
When an item is tagged, positive probabilities are assigned
to transitioning from the item to the tag and vice versa.
Our objective is to choose the subset of $k$ tags which will
maximize traffic to that item, namely the probability of
a random walk reaching the item at some point.
Intuitively, tagging an item causes probability to flow
from the tag to the item, on account of other items with this tag.
Our goal is hence to `leach' as much probability mass as possible
from the system as a whole.


As mentioned, we are interested in maximizing the probability of
a random walk reaching the state corresponding to a new item.
Although this measure resembles the notion of the probability
of an item under a stationary distribution
(on which the popular PageRank metric is based),
it is in fact quite different.
First, while a state's stationary probability increases with incoming links,
it may decrease with outgoing links.
Since assigning a tag results in the addition of both an incoming
\emph{and} an outgoing link, using the stationary probability
of an item will lead to an undesired non-monotone objective
\cite{avrachenkov2006effect}.
Second, a stationary distribution does not always exist,
and hence may require modifications of the Markov chain.
Finally, we argue that maximizing the probability of reaching
an item, as opposed to maximizing the proportion of time an infinite random
walk spends in the item's state,
is a better suited objective for the applications we consider.

Although the Markov chain model we propose for optimal tagging
is bipartite, our results apply to general Markov chains.
We therefore first formulate a general problem
in \secref{sec:problem}, where the task is to choose $k$ states to link a new
state to such that the probability of reaching that state is maximal.
Then, in \secref{sec:hardness} we prove that this problem is NP-hard
by a reduction from vertex cover.
In \secref{sec:submod} we prove that for a general Markov chain the
optimal objective is both monotonically non-decreasing and submodular.
Based on this, in \secref{sec:optimization} we suggest a basic greedy
$(1-\frac{1}{e})$-approximation algorithm, and offer a way
of significantly improving its runtime.
In \secref{sec:optimal_tagging} we revisit the optimal tagging problem 
and show how to construct
a bipartite Markov chain for a given tag-supporting information system.
In \secref{sec:experiments} we present experimental results on
three real-world tagging datasets (musical artists in Last.fm,
bookmarks in Delicious, and movies in Movielens)
and show that our algorithm outperforms plausible heuristics.
Concluding remarks are given in \secref{sec:conclusions}.


\section{Related Work}
\label{sec:related}

One the main roles of tags is to
aid in the categorization and classification of content.
The hallmark of tags is that they are not constrained to a fixed vocabulary or structure.
%
An active line of research in tagging systems focuses on the task
of tag recommendations, where the goal is to predict
the set of tags a given user may attribute to an item.
This setting is especially useful for collaborative tagging systems
and folksonomies, where any user can tag any item.
%
Popular methods for tag recommendation are based on random walks \cite{hotho2006folkrank} or
tensor factorization \cite{fang2015personalized}.
While the goal in tag recommendation is also to output a set of tags,
our task is very different in nature.
Tag recommendation is a prediction task for item-user pairs, is based on ground-truth evaluation, and target collaborative tagging systems.
In contrast, ours is an item-centric optimization task for tag-based taxonomies, and is counterfactual in nature.
As such, tag recommendation methods do not apply to our setting. \todo{BiFolkRank?}

A line of work similar to ours is optimizing the PageRank of web pages
in different settings.
In \cite{csaji2010pagerank} the authors consider the problem of
computing the maximal and minimal PageRank value for a set of ``fragile'' links.
The authors of \cite{avrachenkov2006effect} analyze the effects of additional
outgoing links on the PageRank value.
Perhaps the work most closely related to ours is \cite{olsen2010constant},
where a constant-factor approximation algorithm is given for the problem
of maximizing the PageRank value by adding at most $k$ incoming links.
The authors prove that the probability of reaching a web page
is submodular and monotone in a fashion similar to ours (but with a different parameterization),
and use it as a proxy for PageRank.

Links between absorbing Markov chains and submodular optimization
have been studied for 
opinion maximization \cite{gionis2013opinion}
and for computing centrality measures \cite{mavroforakis2015absorbing}
in networks.
Following the classic work of Nemhauser \cite{nemhauser1978analysis},
submodular optimization is now a very active line of research.
Many interesting optimization problems across diverse domains have been
shown to be submodular.
Examples of these are
sensor placement \cite{krause2008near} and
influence maximization in social networks \cite{kempe2003maximizing},
to name a few.



\section{Problem Formulation}
\label{sec:problem}

Before we present our approach to optimal tagging,
we first describe a general ocombinatorial ptimization task over Markov chains,
of which our task is a special case.
Consider a Markov chain over $n+1$ states.
Assume there is a state $\sigma$ for which we would like to add a set of $k$ new incoming transitions. In the tagging
problem $\sigma=n+1$ will be an item (e.g., song or product) and the incoming transitions will be from possible tags
for the item, or from related items.

The optimization problem is then to choose a subset $\sub \subseteq [n]$ of $k$ states so as to maximize the probability of visiting $\sigma$ at some point in time.
Formally, let $X_t$ be the random variable corresponding to the state of the Markov chain at time $t$.
Then the optimal tagging problem is:
\begin{align}
\max_{S \in [n],\,|S| \le k} \probargsub{X_t=\sigma \text{ for some } t \ge 0}{\sub}
\label{eq:objective}
\end{align}
At first glance, it is not clear how to compute the objective function
in \eqref{eq:objective}. However, with a slight modification of the Markov chain,
the objective function can be expressed as a simple function of the Markov chain parameters, as explained next.

In general, $\sigma$ may have outgoing edges, and random walks
reaching $\sigma$ may continue to other states afterward.
Nonetheless, as we are only interested in the probability of \emph{reaching} $\sigma$,
the states visited after $\sigma$ have no effect on our objective.
Hence, $\sigma$'s outgoing edges can be safely replaced with a single self-edge
without affecting the probability of reaching $\sigma$.
This essentially makes $\sigma$ an \emph{absorbing state},
and our task becomes to maximize the probability of the Markov
chain being absorbed in $\sigma$.
In the remainder of the paper we consider this equivalent formulation.

When the Markov chain includes other absorbing states, 
optimizing over $\sub$ can be intuitively thought of as trying
to ``leach'' as much probability mass from the contending absorbing states
to $\sigma$, under a budget on the number of states $\sigma$
can be connected to.\footnote{
In an ergodic chain with one absorbing state,
all walks reach $\sigma$ w.p. 1, and the problem becomes trivial.}
As we discuss in Section \ref{sec:optimal_tagging},
having contending absorbing states
arises naturally in optimal tagging.

To fully specify the problem, we need the Markov chain parameters. Denote the initial distribution by $\piv$.
 For the transition probabilities, each node $i$ will have two sets of transitions: one when it is allowed to transition to $\sigma$ (i.e., $i \in \sub$) and one when 
 no transition is allowed.
Using two distinct sets is necessary since in both cases outgoing probabilities
must sum to one.
We use $\wsig_{ij}$ to denote the transition probability from state $i$ to $j$ when transition to $\sigma$ is allowed,
and $\nosig_{ij}$ when it is not.
We also denote the corresponding transition matrices by $\Wsig$ and $\Nosig$.
Note that $\nosig_{i \sigma}=0$ for all $i \in [n]$.

It is natural to assume that when adding $\sigma$, transition into $\sigma$ will become more likely, and transition to other states can only be less likely. Thus, we add the assumption that:
\be
\forall i \ , \ \forall j \neq \sigma \ : \ \wsig_{ij} \leq \nosig_{ij}
\label{eq:ge_assumption}
\ee  

Given a subset $S$ of states from which transitions to $\sigma$ are allowed,
we construct a new transition matrix, taking corresponding rows from
$\Wsig$ and $\Nosig$.
We denote this matrix by $\rho(S)$, with
\be
\rho_{ij}(S) = 
\left\{
\begin{array}{ll}
\wsig_{ij} & i \in S \\
\nosig_{ij} & i \notin S \\
\end{array}
\right.
\label{eq:rho_q}
\ee

In what follows, we focus on the optimization problem in \eqref{eq:objective}. \secref{sec:hardness} shows that it is NP hard. \secref{sec:submod} then shows that the objective of \eqref{eq:objective} is monotone and submodular and therefore the optimization problem has a $1-\frac{1}{e}$ factor approximation via a simple greedy algorithm. 

\ignore{
When a state $i$ is chosen to be in $S$,
the transition probability $\rho_{i\sigma}$ 
from $i$ to $\sigma$ can only increase (or remain unchanged).
As such an increase can only come at the expense of other states,
all other non-negative transition probabilities $\rho_{ij}$
to states $j \in [n]$ must decrease.
To model this, we assume $\rho$ describes the transition probabilities for when
all edges from states $i \in [n]$ to $\sigma$ exist.
Hence, when state $i$ is \emph{not} chosen to be in $\sub$,
after normalization the $i^{\text{th}}$ row in $\rho$ becomes
$\frac{\rho_{ij}}{1-\rho_{i\sigma}}$,
since the sum of that row is reduced from 1 to $(1-\rho_{i \sigma})$.\\
\todo{generalize - have $\bar{A}$ be without $\sigma$,
and $A$ with $\sigma$ (or other way around if it makes more sense)}\\
\todo{cascade click model and others}
}
\ignore{
While the transition probability $\rho_{i \sigma}$ increases when $i \in S$,
it is not immediately clear what effect this has on
the probability of reaching $\sigma$.
This is because all walks to $\sigma$ passing through an edge $(i,j)$
for some $j \neq \sigma$ can potentially have lower probabilities due
to a reduced value of $\rho_{ij}$.
In section Sec. \ref{subsec:monotonicity} we prove that adding states to $S$
can only increase the probability of reaching $\sigma$,
or in other words, that the objective function is non-decreasing in $S$.\\

Even thought the objective function can be computed efficiently,
In the next section we show that its optimization is NP-hard.

}


\section{NP-Hardness}
\label{sec:hardness}

We now show that for a general Markov chain, the optimal tagging problem in \eqref{eq:objective} 
is NP-hard by a reduction from vertex cover. Given an undirected graph $G=(V,E)$ with $n$ nodes
as input to the vertex cover problem, we construct an instance of optimal tagging
such that there exists a vertex cover $S \subseteq V$ of size at most $k$ iff
the probability of reaching $\sigma$ reaches some threshold.

To create the absorbing Markov chain,
we create a transient state $i$ for every node $i \in V$,
and add two absorbing states $\varnothing$ and $\sigma$.
We set the initial distribution to be uniform, and for some $0<\epsilon <1$ set the transitions for transient states $i$ as follows:
\be
\wsig_{ij} = 
\left\{
\begin{array}{ll}
1 & j = \sigma \\
0 & j \neq \sigma
\end{array}
\right.
, \qquad
\nosig_{ij} = 
\left\{
\begin{array}{cl}
0 & j = \sigma \\
\epsilon & j = \varnothing \\
\frac{1-\epsilon}{deg(i)} & {\mbox{otherwise}}
\end{array}
\right.
\ee

\ignore{
\emph{unnormalized} transition matrix $\tilde{\rho}$:
for all $i,j \in V$ we set $\tilde{\rho}_{ij} = \frac{1-\epsilon}{d_i}$,
where $d_i$ is $i's$ degree in $G$ and $0 < \epsilon <1$ is chosen arbitrarily,
and $\tilde{\rho}_{i \varnothing}=\epsilon$ and $\tilde{\rho}_{i \sigma} = \infty$\footnote{
To be precise, since the problem is discrete we can set
$\tilde{\rho}_{i \sigma} = \frac{1}{\delta}$ for a sufficiently small $\delta$ for which the argmax is the same.}. \red{in the greedy algorithm we need to invert $\rho$ or solve equations, but with $\infty$ the corresponding row is all zeros and $\rho$ becomes singular - does this matter?}
Given $\sub$, the normalized transition matrix $\rho$ will have $\rho_{ij}=\tilde{\rho}_{ij}$,
$\rho_{i \varnothing} = \epsilon$, and $\rho_{i \sigma}=0$ if $i \not \in \sub$,
and $\rho_{ij} = \rho_{i \varnothing} = 0$ and $\rho_{i \sigma}=1$ if $i \in \sub$.
Finally, we set $\pi_0$ to be uniform.
}
Let $U \subseteq V$ of size $k$, and $\sub (U)$ the set of states corresponding to the nodes in $U$.
We claim that $U$ is a vertex cover in $G$ iff the probability of reaching $\sigma$ when $\sub (U)$ is chosen
is $1-\frac{(n-k)}{n}\epsilon$.

Assume $U$ is a vertex cover. For every $i \in \sub(U)$, as noted a walk starting in $i$ will reach $\sigma$
with probability 1. For every $i \not \in \sub(U)$, with probability $\epsilon$
a walk will reach $\varnothing$ in one step, and with probability $1-\epsilon$ it will visit
one of its neighbors $j$. Since $U$ is a vertex cover, it will then reach $\sigma$ in
one step with probability 1. Hence, in total it will reach $\sigma$ with probability $1-\epsilon$.
Overall, the probability of reaching $\sigma$
is $\frac{k+(n-k)(1-\epsilon)}{n} = 1-\frac{(n-k)}{n}\epsilon$ as needed.
Note that this is the maximal possible probability of reaching $\sigma$ for \emph{any}
subset of $V$ of size $k$.

Assume now that $U$ is not a vertex cover, then there exists an edge $(i,j) \in E$
such that both $i \not \in \sub (U)$ and $j \not \in \sub (U)$.
A walk starting in $i$ will reach $\varnothing$ in one step with probability $\epsilon$,
and in two steps (via $j$) with probability $\epsilon \cdotp \nosig_{ij} > 0$. Hence,
it will reach $\sigma$ with probability strictly smaller than $1-\epsilon$,
and the overall probability of reaching $\epsilon$ will be strictly smaller than $1-\frac{(n-k)}{n}\epsilon$.


\section{Proof of Monotonicity and Submodularity} 
\label{sec:submod}

Denote by $\probargsub{A}{\sub}$ the probability of event $A$
when transitions from $S$ to $\sigma$ are allowed. We define:
\begin{flalign}
& c_i^\k(\sub) = \probargsub{\text{$X_t = \sigma$ for some $t\leq k$} | X_0 = i}{\sub} \\
& c_i(\sub) = \probargsub{\text{$X_t = \sigma$ for some $t$} | X_0 = i}{\sub} 
= \lim\nolimits_{k \rightarrow\infty} c_i^\k 
\label{eq:cis} 
\end{flalign}
For $\cc(\sub) = \left( c_1(\sub), \dots, c_n(\sub) \right)$,
the objective in \eqref{eq:objective} now becomes:
\begin{align}
\max_{\sub \subseteq [n], |\sub| \le k} f(\sub), \qquad f(\sub) = 
\inner{\piv,\cc(\sub)} =
\probargsub{X_t = \sigma \text{ for some $t$}}{\sub}
\label{eq:opttag}
\end{align}
We now prove that $f(\sub)$ 
is both monotonically non-decreasing and submodular.

\subsection{Monotonicity}
\label{subsec:monotonicity}
When a link is from $i$ $\sigma$, the probability of reaching it directly from $i$ goes up.
However, due to the renormalization constraints, the probability of reaching it via other paths may go down.
Nonetheless, our proof of monotonicity shows that the overall probability cannot decrease, as stated next.

\ignore{
As more links are added into $\sigma$, it seems intuitive that the probability of reaching it should increase, and thus $f(S)$ should be non-decreasing.
This however is not trivial, since
adding a new state $i$ to $\sub$ has two distinct effects.
The immediate effect is that a random walk can now reach $\sigma$
in one step from $i$, thus increasing $\sigma$'s traffic volume.
On the other hand, due to renormalization, a walk reaching $i$
has a \emph{lower} probability of visiting other states,
and specifically those in $\sub$.
The indirect effect of this is that 
the probability of reaching $\sigma$ via longer paths 
can actually decrease.
Nonetheless, our proof of monotonicity
shows that the overall probability
cannot decrease, as stated next.
}

\begin{theorem}
For every $k\ge0$ and $i \in [n]$, $c_i^\k$ is non-decreasing.
Namely, for all $\sub \subseteq [n]$ and $z \in [n] \setminus \sub$,
it holds that $c_i^\k(\sub) \le c_i^\k(\sub \cup \{z\})$.
\end{theorem}

\begin{proof}

\ignore{
For clarity we omit the set notations and denote
$\bar{c}_i^\k=c_i^\k(\sub)$ and $c_i^\k=c_i^\k(\sub \cup \{z\})$,
and similarly $\bar{\rho} = \rho(\sub)$ and $\rho = \rho(\sub)$.
Notice that for all $j$ and any $i \neq z$, $\bar{\rho}_{ij}=\rho_{ij}$.
For $i=z$, $\bar{\rho}_{zj}=\frac{\rho_{zj}}{1-\rho_{z \sigma}}$.
due to the renormalization constraints.
}

We prove by induction on $k$.
For $k=0$, as $\piv$ is independent of $\sub$ and $z$, we have:
\[
c_i^0(\sub) = \piv_\sigma  \1{i=\sigma} = c_i^0(\sub \cup \{z\})
\]
Assume now that the claim holds for some $k \ge 0$. We separate into cases. When $i \neq z$, we have:
\bea
c_i^\kk(S) = \sum_{j=1}^n \wsig_{ij} c_j^\k(S) + \wsig_{i \sigma} 
&\leq& \sum_{j=1}^n \wsig_{ij} c_j^\k(S \cup z ) + \wsig_{i \sigma} = c_i^{\kk}(S \cup z) \\
c_i^\kk(S) = \sum_{j=1}^n \nosig_{ij} c_j^\k(S) 
&\leq& \sum_{j=1}^n \nosig_{ij} c_j^\k(S \cup z )  = c_i^{\kk}(S \cup z)
\eea
for $i \in \sub$ and $i \notin \sub$, respectively.
If $i = z$ then:
\bea
c_i^\kk(S) 
 &\leq& \sum_{j=1}^n \nosig_{ij} c_j^\k(S \cup z )  
 = \sum_{j=1}^n \wsig_{ij} c_j^\k(S \cup z ) +  \sum_{j=1}^n (\nosig_{ij} - \wsig_{ij}) c_j^\k(S \cup z )  \\
& \leq& \sum_{j=1}^n \wsig_{ij} c_j^\k(S \cup z) + \sum_{j=1}^n (\nosig_{ij} - \wsig_{ij})  
 = \sum_{j=1}^n \wsig_{ij} c_j^\k(S \cup z ) + \wsig_{z \sigma}  = c_i^{\kk}(S \cup z) 
\eea
due to to $\nosig_{ij} \geq \wsig_{ij}$, $c \leq 1$, 
$\sum_{j=1}^n \nosig_{ij} = 1$, and $\sum_{j=1}^n \wsig_{ij} = 1 - \wsig_{i\sigma}$.
\end{proof}

\begin{corollary}
\label{cor:non_dec}
$\forall i \in [n]$, $c_i(\sub)$ is non-decreasing, hence
$f(S) = \inner{\piv,\cc(\sub)}$ is non-decreasing.
\end{corollary}


\subsection{Submodularity}
\label{subsec:submodularity}

Submodularity captures the principle of diminishing returns. 
A function $f(S)$ is submodular if:
\be
\forall \sub \subseteq [n], \,\, z_1,z_2 \in [n] \setminus \sub, \quad
f(\sub \cup \{z_1\}) + f(\sub \cup \{z_2\}) \geq
f(\sub \cup \{z_1, z_2\}) + f(\sub) \nonumber
\label{eq:submod2}
\ee
It turns out that $f(S)$ as defined in \eqref{eq:opttag} is submodular, as shown in following theorem and corollary.
\begin{theorem}
For every $k\ge0$ and $i \in [n]$, $c_i^\k(S)$ is a submodular function.
\end{theorem}

\begin{proof}
We prove by induction on $k$. 
The case for $k=0$ is trivial since $\piv$ is independent of $\sub$ and hence $c_i^0$ is modular.
Assume now that the claim holds for some $k\ge0$. For brevity we define:
\begin{align*}
c_{i}^\k = c_i^\k(\sub), \quad
c_{i,1}^\k = c_i^\k(\sub \cup \{z_1\}), \quad
c_{i,2}^\k = c_i^\k(\sub \cup \{z_2\}), \quad
c_{i,12}^\k = c_i^\k(\sub \cup \{z_1,z_2\})
\end{align*}
For any $T \subseteq [n]$, we have:
\be
c_i^{(k+1)}(T) = \sum_{j=1}^n \rho_{ij}(T) c_j^{(k)}(T) + \rho_{i \sigma} \1{i \in T}
\label{eq:csumt}
\ee
We'd like to show that
$c_{i,12}^\kk+  c_i^\kk \leq c_{i,1}^\kk +  c_{i,2}^\kk$. For every $j \in [n]$, we'll prove that:
\be
\rho_{ij}(\sub \cup \{z_1, z_2\}) c_{j,12}^\k + 
\rho_{ij}(\sub) c_{j}^\k  \le
\rho_{ij}(\sub \cup \{z_1\}) c_{j,1}^\k +
\rho_{ij}(\sub \cup \{z_2\}) c_{j,2}^\k 
\label{eq:submod_ineq_j}
\ee
%
%
%
which together with \eqref{eq:csumt} and subtracting $\rho_{i \sigma} \1{i \in T}$ from both sides
will conclude our proof. We separate into different cases for $i$.
If $i \in S$, then we have $\rho_{ij}( \sub \cup \{z_1,z_2\}) = \rho_{ij}( \sub \cup \{z_1\})   = \rho_{ij}( \sub \cup \{z_2\})  = \rho_{ij}( \sub) = \wsig_{ij}$. 
Similarly, if $i \notin S\cup \{z_1,z_2\}$, 
then all terms now equal $\nosig_{ij}$.
  Therefore by the inductive assumption \eqref{eq:submod_ineq_j} follows.
Assume $i=z_1=z$ (and analogously $i=z_2$).
From the assumption in \eqref{eq:ge_assumption}
we can write $\nosig_{ij} = (1+\alpha) \wsig_{ij}$ for some $\alpha\geq 0$.
Then \eqref{eq:submod_ineq_j} becomes:
\be
\wsig_{ij} c_{j,12}^\k + (1+\alpha)\wsig_{ij} c_{j}^\k \le
\wsig_{ij}  c_{j,1}^\k +  (1+\alpha)\wsig_{ij} c_{j,2}^\k 
\ee
Reorder to get:
\be
c_{j,1}^\k + c_{j,2}^k  - c_{j,12}^\k - c_{j}^\k + \alpha (c_{j,2}^{k} - c_{j}^\k) \geq 0 
\ee
This indeed holds since the first four terms are non-negative from the inductive assumption, and the last term is non-negative
because of monotonicity and $\alpha\geq 0 $.
\end{proof}

\begin{corollary}
\label{cor:submod}
$\forall \,i \in [n]$, $c_i(\sub)$ is submodular, hence $f(\sub)=\inner{\piv,\cc(\sub)}$ is submodular.
\end{corollary}
%


\section{Optimization}
\label{sec:optimization}
Maximizing submodular functions is hard in general.
However, a classic result by Nemhauser \cite{nemhauser1978analysis} shows that a
non-decreasing submodular set function, such as our $f(S)$, can be efficiently
optimized via a simple greedy algorithm,
with a guaranteed $(1-\frac{1}{e})$-approximation of the optimum.
The greedy algorithm initializes $\sub = \emptyset$, and then sequentially adds elements to $\sub$.
For a given $\sub$, the algorithm iterates over all $z \in [n] \setminus \sub$
and computes $f(\sub \cup \{z\})$.
Then, it adds the highest scoring $z$ to $\sub$, and continues to the next step.
We now discuss its implementation for our problem.

\begin{algorithm}[t]
\caption{}
\label{array-sum}
\begin{algorithmic}[1]
\Function{SimpleGreedyTagOpt}{$\Wsig,\Nosig,\piv,k$}
		 \Comment{See supp. for efficient implementation}
  \State \text{Initialize } $\sub = \emptyset$
  \For {$i \gets 1$ \text{ to } $k$}
    \For {$z \in [n] \setminus \sub$} 
            \State $\cc = \big(I - A(\sub \cup \{z\}) \big) \setminus  \, \bb(\sub \cup \{z\}) $
	            \Comment{$A,\bb$ are set using Eqs. (\ref{eq:rho_q}), (\ref{eq:rho_AB})}  
      \State $v(z) = \inner{\piv,\cc}$
    \EndFor
    \State $\sub \gets \sub \cup \argmax_z v(z)$  
  \EndFor
  \State Return $\sub$
\EndFunction
\end{algorithmic}
\label{algo:greedy}
\end{algorithm}

Computing $f(\sub)$ for a given $\sub$ reduces to solving a set of linear equations.
For a Markov chain with transient states $\{1,\dots,n-r\}$ and absorbing states $\{n-1+1,\dots,n+1=\sigma\}$,
the transition matrix $\rho(\sub)$ can be written as:
  \begin{equation}
\rho(\sub) = \left( \begin{array}{cc}
A(\sub) & B(\sub)\\
\bm{0} & I
\end{array} \right)
\label{eq:rho_AB}
\end{equation}
where $A(\sub)$ are the transition probabilities between transient states,
$B(\sub)$ are the transition probabilities from transient states to absorbing states,
and $I$ is the identity matrix. 
When clear from context we will drop the dependence of $A,B$ on $\sub$. 
Note that $\rho(\sub)$ has at least one absorbing state (namely $\sigma$).
We denote by $\bb$ the column of $B$ corresponding to state $\sigma$ (i.e., $B$'s rightmost column).	

We would like to calculate $f(\sub)$.
From \eqref{eq:cis},
the probability of reaching $\sigma$ given initial state $i$ is:
\ignore{
The probability of a random walk starting in a transient state $i$ and ending in
an absorbing state $k$ is $C_{ik}$, where:
\[
C=(I-A)^{-1}B
\]
}
\begin{align*}
c_i(S) = 
\sum_{t=0}^\infty \sum_{j \in [n-r]}
\probargsub{X_t=\sigma | X_{t-1} = j}{S}
\probargsub{X_{t-1}=j | X_{0} = i}{S} 
= \left( \sum_{t=0}^\infty A^t \bb \right)_{i}
\end{align*}
The above series has a closed form solution:
\be
\sum_{t=0}^\infty A^t = (I-A)^{-1}
\quad \Rightarrow \quad
\cc(S) = (I-A)^{-1} \bb \nonumber
\ee
Thus, $\cc(S)$ is the solution of the set of linear equations, whose solution readily gives us $f(\sub)$:
\be
f(S) = \inner{\piv,\cc(S)}
\quad \text{s.t.} \quad
(I-A) \cc(S) = \bb
\label{eq:f_objective}
\ee
\ignore{
As we are interested in the probability of reaching $\sigma$, it suffices to focus on:
\[
c=(I-A)^{-1} b, \qquad b_i = B_{i \sigma} \,\, \forall i
\]
Note that due to the normalization constraints, $A$ and $b$ (and hence $c$) are in fact
functions of $\sub$.
Finally, to compute $f(\sub)$, we solve the set of linear equations
and take an inner product with $\piv$:

\begin{equation}
f(\sub) = \inner{\piv,\cc(S)} \quad \text{s.t.} \quad  (I-A(\sub))c = b(\sub)
\label{eq:f_objective}
\end{equation}
}

The greedy algorithm can thus be implemented by 
sequentially considering candidate sets $S$ of increasing size,
and for each $z$ calculating $f(\sub \cup \{z\})$ 
by solving a set of linear equations (see Algo. \ref{algo:greedy}).
Though parallelizable, this na\"{\i}ve implementation may be quite costly as it requires solving $O(n^2)$ sets of $n-r$ linear equations,
 one for every addition of $z$ to $\sub$.
Fast submodular solvers such as CELF++ \cite{goyal2011celf++}
can reduce the number of calls to $f(\sub)$ by an order of magnitude.
As we now show, a significant speedup in the computation of $f(\sub)$ itself can be achieved using problem's structure.

A standard method for solving the set of linear equations $(I-A)\cc=\bb$ if to first
compute an $LUP$ decomposition for $(I-A)$, namely find lower and upper diagonal matrices $L,U$
and a permutation matrix $P$ such that
$LU = P(I-A)$. Then, if suffices to solve $Ly=Pb$ and $Uc=y$.
Since $L$ and $U$ are diagonal, solving these equations can be performed efficiently;
the costly operation is computing the decomposition in the first place.
Recall that $\rho(\sub)$ is composed of rows from $\Wsig$ corresponding to $\sub$
and rows from $\Nosig$ corresponding to $[n] \setminus \sub$.
This means that $\rho(\sub)$ and $\rho(\sub \cup \{z\})$ differ only in one row,
or equivalently, that $\rho(\sub \cup \{z\})$ can be obtained from $\rho(\sub)$
by adding a rank-1 matrix.

With this in mind, given an $LUP$ decomposition of $\rho(\sub)$, we can efficiently compute
$f(\sub \cup \{z\})$ (and its corresponding decomposition)
using rank-1-update techniques such as Bartels-Golub-Reid \cite{reid1982sparsity} or others.
Such methods are especially efficient for sparse matrices.
As a result,
it suffices to compute
only a \emph{single} $LUP$ decomposition for the input at the beginning,
and perform cheap updates at every step.
See the supplementary material for an efficient implementation.

\ignore{
\subsection{Efficient Implementation}
Algorithm \ref{algo:greedy} presents a na\"{\i}ve implementation of the greedy approach to submodular maximization, in which
$f(\sub)$ needs to be computed $O(nk)$ times.
A simple way achieve considerable speedup is via parallelization,
where for a given set $\sub$, evaluations
of $f(\sub \cup \{z\})$ for all $z \in [n] \setminus \sub$
are performed in parallel.

Parallelization can reduce running time, but still requires $O(nk)$ evaluations
of $f(\sub)$.
The literature contains several methods where the submodularity of $f$ is used to achieve a considerable reduction in the number of computations.
For instance, the CELF \cite{leskovec2007cost} algorithm for submodular
maximization maintains a list of lazy evaluations of the marginal gains
$f(\sub \cup \{z\}) - f(\sub)$, sorted by non-decreasing values.
The submodularity of $f$ ensures that marginal gains can only decrease
as $\sub$ grows. Hence, a lazy evaluation serves as an upper bound
on the true marginal gain.
The algorithm repeatedly updates the marginal gain of the state at the top
of the list. If after the update the state maintains its rank,
then it is guaranteed to be the correct greedy choice.
Otherwise, the list is resorted.
In our experiments, using CELF resulted in roughly $2n'$ evaluations
of $f(\sub)$ on average for $k=25$, where $n'<n$ is the number
of states with non-zero values in $\bb$.
Other pupolar methods include CELF++ \cite{goyal2011celf++} and UBLF \cite{zhou2013ublf}.

A straightforward way to compute $f(\sub)$ is to solve the set
of linear equations in Eq. \ref{eq:lin_eqs}.
Nonetheless, the total number of computations can be significantly reduced
due to the fact that
for any $\sub$ and $z \not \in \sub$,
the matrices $A(\sub)$ and $A(\sub \cup \{z\})$ 
differ only in one row.
\ignore{
\red{Specifically, the $z^{\text{th}}$ row of $A(\sub \cup \{z\})$ is just
the corresponding row in $A(\sub)$ multiplied by $(1-b_z)$,
and while $(b(\sub))_z = 0$, $(b(\sub \cup \{z\}))_z = \rho_{z \sigma}$.}
}

A standard method for solving the set of linear equations $(I-A)\cc=\bb$ if to first
compute an $LUP$ decomposition for $(I-A)$, namely find lower and upper diagonal matrices $L,U$
and a permutation matrix $P$ such that
$LU = P(I-A)$. Then, if suffices to solve $Ly=Pb$ and $Uc=y$.
Since $L$ and $U$ are diagonal, solving these equations can be performed
efficiently with forward and backward substitution;
the costly operation is computing the decomposition in the first place.

Since $A(\sub)$ and $A(\sub \cup \{z\})$ differ only in the $z^{\text{th}}$ row,
so do $(I-A(\sub))$ and $(I-A(\sub \cup \{z\}))$.
As we show next, given an $LUP$ decomposition of $(I-A(\sub))$, a decomposition
of $(I-A(\sub \cup \{z\}))$ can be computed efficiently.

To see why, note that:
\begin{align*}
(LU)_{\Pi_i \cdotp} =
 P_{\Pi_i \cdotp}(I-A(\sub)) =
\begin{cases}
\ev_i - \Nosig_{i \cdotp} &\mbox{if } i \neq z \\ 
\ev_i - \Wsig_{i \cdotp} &\mbox{if } i = z 
\end{cases} 
\end{align*}
where $\Pi$ is the permutation given by $P$, and $\ev_i$ is an indicator vector.
A similar expression for $\sub \cup \{z\}$ differs only in the $\Pi_z^{\text{th}}$ row:
\begin{align*}
P_{\Pi_z \cdotp}(I-A(\sub \cup \{z\})) &=
\ev_z - \Wsig_{z \cdotp} \\
&= \ev_z - \Nosig_{z \cdotp} + (\Nosig_{z \cdotp} - \Wsig_{z \cdotp})  \\
&= P_{\Pi_z \cdotp}\left[ (I-A(\sub)) +\ev_z (\Nosig_{z \cdotp} - \Wsig_{z \cdotp})^\top \right]
\end{align*}
This offers a new decomposition:
\begin{align*}
L'U = P (I-A(\sub \cup \{z\})), \quad
L' = L + \ev_z (\boldsymbol{\ell}^\top - L_{\Pi_z \cdot})
\end{align*}
Where $\boldsymbol{\ell}$ is the solution to
$\boldsymbol{\ell}^\top U = \Nosig_{z \cdot}$.
Since $L'$ is obtained from $L$ by an addition of a rank-1 matrix,
the LUP decomposition of $(I-A(\sub \cup \{z\}))$ can by computed
efficiently \todo{numbers!}
using rank-1 update methods such as Bartels-Golub-Reid \cite{reid1982sparsity}
or Forrest-Tomlin \cite{forrest1972updated}.
These methods build on the fact that while $L'$ is not triangular,
it can be efficiently transformed into
a lower triangular matrix by a small number of row or column swaps.\footnote{
The efficient LUSOL package for sparse rank-1 $LU$ modifications
can be found at \url{https://web.stanford.edu/group/SOL/software/lusol/}.}
The above methods are especially efficient
for sparse matrices.

As a result,
in both the na\"{i}ve implementation and in CELF,
it suffices to compute
a single $LUP$ decomposition for the input at the beginning,
and perform cheap updates at every step.
} 

\section{Optimal Tagging}
\label{sec:optimal_tagging}


In this section we return to the task of optimal tagging
and show how the Markov chain optimization framework described
above can be applied.
We use a random surfer model, where a browsing user hops between
items and tags in a bipartite Markov chain.
In its explicit form, our model captures the activity of browsing users whom,
when viewing an item,
are presented with the item's tags and may choose to click on them
(and similarly when viewing tags).

In reality, many systems also (or mainly) include direct links between related items,
often in the form of a ranked list of item recommendations.
The relatedness of two items is often, at least to some extent,
based on the set of mutual tags.
Our model captures this notion of similarity by the (implicit)
transitions to and from tag states.
This allows us to encode tags as variables in the objective.
As our results apply to general Markov chains,
adding direct transitions between items is straightforward.
\blue{
We further elaborate on this point of view in relation to parameter estimation in XXX.
}
Our framework easily incorporates click models \cite{chuklin2015click},
in which adding an item to a rank list modifies the clicking probabilities.
Note that in contrast to models for tag recommendations, we do not need to explicitly model the system's users,
as our setup defines only one distinct optimization task per item.

\ignore{
In many systems, when a browsing user queries a tag or searches for a keyword,
she is presented with a list of items ranked by their popularity, relevance,
or their trending status.
When an item is tagged, it is included in this list, and can potentially
effect the rank of other items.
With respect to \eqref{eq:ge_assumption},
our only assumption here is that while adding a tag $\tg_j$ to $\sub$
increases $\rr_j$, each entry in $R_{j \cdot}$ can either
decrease or remain unchanged
(and that $\rho_{j \cdot}$ must remain row-normalized).
This condition holds for simple renormalization,
as well as for many web search click models
\cite{chuklin2015click}.
}

In what follows we formalize the above notions.
Consider a system of $m$ items $\itms=\{\itm_1,\dots,\itm_m\}$
and $n$ tags $\tgs=\{\tg_1,\dots,\tg_n\}$.
Each item $\itm_i$ has a set of tags $\tgs_i \subseteq \tgs$,
and each tag $\tg_j$ has a set of items $\itms_j \subseteq \itms$.
The items and tags constitute the states of a bipartite Markov chain,
where users hop between items and tags.
Specifically, the transition matrix $\rho$ can have non-zero
entries $\rho_{ij}$ and $\rho_{ji}$ for items $\itm_i$ tagged by $\tg_j$.
To model the fact that browsing users eventually leave the system,
we add a global absorbing state $\varnothing$ and add transition probabilities
$\rho_{i \varnothing}=\epsilon_i>0$ for all items $\itm_i$.
For simplicity we assume that $\epsilon_i=\epsilon$ for all $i$,
and that $\piv$ can be non-zero only for tag states.

In our setting, when a new item $\sigma$ is uploaded,
the uploader may choose a set $\sub \subseteq \tgs$ of at most $k$ tags for $\sigma$.
Her goal is to choose $\sub$ such that the probability
of an arbitrary browsing user reaching (or equivalently, being absorbed in)
$\sigma$ while browsing the system is maximal.
As in the general case, the choice of $\sub$
effects the transition matrix $\rho(\sub)$.

Denote by $P_{ij}$ the transition probability
from item $\itm_i$ to tag $\tg_j$,
by $R_{ji}(\sub)$ the transition probability from $\tg_j$ to $\itm_i$ under $\sub$,
and let $\rr_j(\sub) = R_{j \sigma}(\sub)$.
Using \eqref{eq:rho_AB}, $\rho$ can be written as:
\be 
\rho(\sub) = \left( \begin{array}{cc}
A & B\\
\bm{0} & I_2
\end{array} \right), \quad \,\,
A= \left( \begin{array}{cc}
\bm{0} & R(\sub)\\
P & \bm{0}
\end{array} \right), \quad \,\,
B = \left( \begin{array}{cc}
\bm{0} & \rr(\sub) \\
\bm{1} \cdotp \epsilon & \bm{0}
\end{array} \right), \quad \,\,
I_2 = \left( \begin{array}{cc}
1 & 0\\
0 & 1
\end{array} \right) \nonumber
\ee 
where $\bm{0}$ and $\bm{1}$ are appropriately sized vectors or matrices.
Since the graph is bipartite,
since we assume the walks start at a tag state,
and since we are interested only in choosing tags,
the Markov chain can be ``folded'' to include only the tag-states and
the absorbing states.
Looking at $\rho^2(\sub)$,
the transition probabilities between tags
are now given by the matrix $R(\sub)P$,
while the transition probabilities from tags to $\sigma$ remain $\rr(\sub)$.
Our objective of the probability of reaching $\sigma$ under $\sub$ is:
\be
f(\sub) = \inner{\piv,\cc(\sub)}  \quad \text{s.t.} \quad
\left( I-R(\sub)P \right) \cc(\sub) = \rr(\sub)
\ee
which is a special case of the general objective presented in \eqref{eq:f_objective},
and hence can be optimized efficiently.
In the supplementary material we prove that this special case is still NP-hard.



\ignore{
As for the computational hardness of optimal tagging,
a reduction from vertex cover similar to that in
\secref{sec:hardness} can be constructed for the above special case.
Given a graph $G=(V,E)$,
we create a bipartite Markov chain with a state $\tg_i$ for every node $i \in V$
and a state $\itm_{ij}$ for every edge $(i,j) \in E$,
and add two absorbing states $\varnothing$ and $\sigma$.
We set the transition probabilities in $\nosig$ to be $\frac{1}{d_i}$
from $\tg_i$ to any $\itm_{ij}$, $(1-\epsilon)$ from $\itm_{ij}$ to $\tg_j$,
and $\epsilon$ from $\itm_{ij}$ to $\varnothing$,
and continue as in the general case.
}


\section{Experiments}
\label{sec:experiments}

\begin{figure*}[t]
  \begin{center}
  \includegraphics[width=\textwidth]{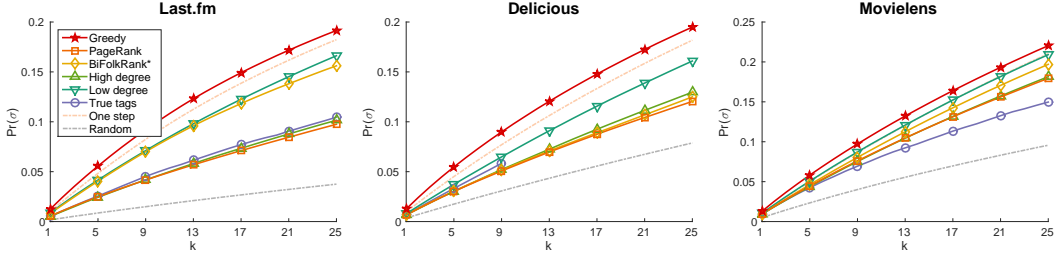}
\end{center}
\caption{The probability of reaching a focal item $\sigma$
under a budget of $k$ tags for various methods.}
\label{fig:results}
\end{figure*}

To demonstrate the effectiveness of our approach,
we perform experiments on optimal tagging
in data collected from Last.fm, Delicious, and Movielens.
taken from the HetRec 2011 workshop \cite{Cantador:RecSys2011}.
The datasets include all items (between 10,197 and 59,226) and tags (between 11,946 and 53,388) reached from crawling
a set of about 2,000 users in each system, as well as some metadata.


For each dataset, we first created a bipartite graph of items and tags.
Next, we generated 100 different instances of our problem per dataset by expanding
each of the 100 highest-degree tags
and creating a Markov chain for their items and all their associate tags.
We discarded nodes with less than 10 edges.

To create an interesting tag selection setup,
for each item in each instance we augmented its true tags
with up to 100 similar tags (based on \cite{sigurbjornsson2008flickr}).
These served as the set of candidate tags for that item.
We focused on items which were ranked first in at least
10 of their 100 candidate tags,
giving a total of 18,167 focal items for comparison.
For each such item, our task was to choose the $k$ tags
out of the 100 candidate tags
which maximize the probability of reaching the focal item.

Transition probabilities from tags to items were set to be proportional
to the item weights - number of listens for artists in Last.fm,
tag counts for bookmarks in Delicious, and averaged ratings for movies in Movielens.
As the datasets do not include weights for tags, we used 
uniform transition probabilities from items to tags.
The initial distribution was set to be uniform over the set
of candidate tags,
and the transition probability from items to
the absorbing state $\varnothing$ was set to $\epsilon=0.1$.



We compared the performance of our greedy algorithm
with several baselines.
Random-walk based methods included PageRank
and a variant\footnote{To apply the method to our setting, we used a uniform prior over user-tag relations.}
 of BiFolkRank \cite{kim2011personalized}, a state-of-the-art tag recommendation method that operates on item-tag relations.
Heuristics included taking the $k$ tags with highest degree, lowest degree,
the true labels (for relevant $k$-s), and random.
To measure the added value of taking into account
long random walks, we also display the probability of reaching $\sigma$
in one step.

Results for all three datasets are provided in Figure \ref{fig:results},
which shows the average probability of reaching the focal item
for values of $k \in \{1,\dots,25\}$.
As can be seen, the greedy method clearly outperforms other baselines.
In our setup choosing low degree tags outperforms both random-walk based methods and heuristics.
Considering paths of all lengths improves results by a considerable 20-30\% for $k=1$,
and roughly 5\% for $k=25$.
%
\ignore{
Figure \ref{fig:ratio} shows the increase in the probability of reaching
$\sigma$ when all path lengths are taken into account in the greedy algorithm,
relative to its first-order approximation,
namely choosing the $k$ tags with the highest probability
of reaching $\sigma$ in one step.
For the initial selection ($k=1$), the increase is rather dramatic
at 20-30\% across datasets.
As $k$ grows, the added effect decreases down to roughly 5\% for $k=25$.
}
An interesting observation is that the performance of the true tags
is rather poor.
A plausible explanation for this is that the data we use
are taken from collaborative tagging systems,
where items can be tagged by any user. 
In such systems, tags typically play a categorical or hierarchical role,
and as such are probably not optimal for promoting item popularity.
\blue{
An analysis of a single test case which qualitatively compares the true tags
and those suggested by the greedy approach is presented in the supplementary material.
}

\ignore{
Figure \ref{fig:greedy_vs_true} displays a test-case analysis
of the tags selected for a focal item by the greedy algorithm,
and the item's true tags ordered by their weight.
In this example, the greedy algorithm chose several categorical tags
(such as ``groove'' and ``synth pop'')
that appear in the set of true tags as well.
The order of these tags in the greedy selection is however opposite
from that of the true tags, suggesting that
being well connected to less popular tags may be more beneficial.
The greedy algorithm also selected several cross-categorical tags
such as ``makes me happy'' and ``one of my favorites''.
While probably not very popular on their own,
due to their wide degree of connectivity they are likely candidates.
}



\section{Conclusions}
\label{sec:conclusions}

In this paper we introduced the problem of optimal tagging,
along with the general problem of optimizing probability mass in Markov chains, by adding links.
We proved that the problem is NP-hard,
but can be $(1-\frac{1}{e})$-approximated
due to the submodularity and monotonicity of the objective.
Our efficient greedy algorithm
can be used in practice for choosing optimal tags
or keywords in various domains.
Our experimental results show that simple heuristics and PageRank variants underperform
our disciplined approach, and na\"{\i}vely selecting the true tags
is typically suboptimal.

In our work we assumed access to the transition probabilities
between tags and items and vice versa.
While the transition probabilities for existing items
can be easily estimated by a system's operator,
estimating the probabilities from tags to \emph{new} items
is non-trivial.
This is an interesting problem to pursue.
Even so, users do not typically have access to 
the information required for estimation.
Our results suggest that 
users can simply apply the greedy steps sequentially via trial-and-error.

Finally,
as our task is of a counterfactual nature,
it is hard to draw conclusions from the experiments as to the effectiveness of our method in real settings.
It would be interesting to test it in realty, and compare it 
to strategies used by both lay users and experts. 
Especially interesting in this context are competitive domains such as ad placements
and viral marketing. We leave this for future research.

\bibliographystyle{plain}
\bibliography{tagopt_NIPS16}

\end{document}